\documentclass[11pt,a4paper]{article}
\pdfoutput=1
\usepackage{cancel}
\usepackage{enumitem}
\usepackage{etex}
\usepackage{amsthm}
\usepackage{geometry}
\usepackage{dcolumn}
\usepackage{epsf}
\usepackage{mathrsfs}
\usepackage{multirow}
\usepackage{booktabs}
\usepackage{tabularx}
\usepackage{array}
\usepackage{slashed}
\usepackage{float}
\usepackage{leftidx}
\usepackage{setspace}
\usepackage{verbatim}
\usepackage{adjustbox}
\usepackage{tikz}
\usepackage[caption=false]{subfig}
\usepackage{arydshln}
\usepackage{jheppub}
\usepackage{shuffle}
\usepackage{amsmath}
\usepackage{cases}

\usepackage[latin1,utf8]{inputenc}
\usepackage[T1]{fontenc}

\numberwithin{equation}{section}
\usetikzlibrary{arrows,decorations.markings,shapes,arrows,patterns,positioning}

\newcommand{\NB}{N}
\newcommand{\NF}{n}
\newcommand{\NL}{N_L}

\theoremstyle{plain}
\newtheorem{theoremcounter}{}[]

\newtheorem{proposition}[theoremcounter]{Proposition}

\newcommand{\bea}{\begin{eqnarray}}
	\newcommand{\eea}{\end{eqnarray}}
\newcommand{\bean}{\begin{eqnarray*}}
	\newcommand{\eean}{\end{eqnarray*}}
\newcommand{\nn}{\nonumber\\}
\newcommand{\Sl}{\sum\limits}

\def\Label#1{\label{#1}}
\renewcommand{\eqref}[1]{eq.~(\ref{#1})}
	\newcommand{\figref}[1]{Fig.~\ref{#1}}

	\newcommand{\appref}[1]{appendix~\ref{#1}}

	\def\eps{\epsilon}

	\def\Sl{\sum\limits}

	\newcommand{\ctobedelete}[1]{}
	\allowdisplaybreaks

\title{Rationalisation of multiple square roots in Feynman integrals}

\author[a,b]{Georgios Papathanasiou,}
\author[c]{Stefan Weinzierl,}
\author[b,d,f]{Konglong Wu}
\author[e,f]{and Yang Zhang}
\affiliation[a]{Department of Mathematics, City, University of London,
Northampton Square, EC1V 0HB, London, UK}
\affiliation[b]{Deutsches Elektronen-Synchrotron DESY, Notkestr. 85, 22607 Hamburg, Germany}
\affiliation[c]{PRISMA Cluster of Excellence, Institut für Physik, Staudinger Weg 7, Johannes Gutenberg-Universität Mainz, D - 55099 Mainz, Germany}
\affiliation[d]{School of Physics and Technology, Wuhan University, No.299 Bayi Road, Wuhan 430072, China}
\affiliation[e]{Interdisciplinary Center for Theoretical Study, University of Science and Technology of China, Hefei, Anhui 230026, China}
\affiliation[f]{Peng Huanwu Center for Fundamental Theory, Hefei, Anhui 230026, China}
\emailAdd{georgios.papathanasiou@desy.de}
\emailAdd{weinzierl@uni-mainz.de}
\emailAdd{konglong.wu@desy.de}
\emailAdd{yzhphy@ustc.edu.cn}

\preprint{
\begin{flushright} DESY-25-001 \\ MITP-25-004 \\ USTC-ICTS/PCFT-24-28
\end{flushright}
}

\abstract{
Feynman integrals are very often computed from their differential equations.
It is not uncommon that the $\varepsilon$-factorised differential equation contains only dlog-forms with algebraic arguments,
where the algebraic part is given by (multiple) square roots.
It is well-known that if all square roots are simultaneously rationalisable, the Feynman integrals can be expressed
in terms of multiple polylogarithms.
This is a sufficient, but not a necessary criterium.
In this paper we investigate weaker requirements.
We discuss under which conditions we may use different rationalisations in different parts of the calculation.
In particular we show that we may use different rationalisations if they correspond to different parameterisations
of the same integration path.
We present a non-trivial example -- the one-loop pentagon function with three adjacent massive external legs involving seven square roots --
where this technique can be used to express the result in terms of multiple polylogarithms.
} 

\date{}
\begin{document}
\maketitle
\flushbottom

\newpage

\section{Introduction}

Feynman integrals play an essential role in precision calculations for particle physics phenomenology.
Of particular importance is the question to which class of functions a particular Feynman integral evaluates.
The simplest class of functions are the multiple polylogarithms.
They are generalisations of the ordinary logarithm $\ln(x)$ and
correspond to iterated integrals on a curve of genus zero.
More complicated Feynman integrals are associated with curves of higher genus or geometries of higher dimension.
The most prominent examples for the latter case are Calabi-Yau geometries.
An important question is to decide, whether a given Feynman integral can be expressed in terms of functions from a given class of functions.
Already for the simplest class of functions, i.e. the multiple polylogarithms, this is a non-trivial question.

Multiple polylogarithms are a well-understood class of functions.
In particular, numerical evaluation routines are available for all complex arguments \cite{Vollinga:2004sn,Maitre:2005uu,Maitre:2007kp,Naterop:2019xaf,Duhr:2019tlz,Wang:2021imw}.
In typical phenomenological applications we need to evaluate Feynman integrals and scattering amplitudes for different values of the arguments many times
(${\mathcal O}(10^6)-{\mathcal O}(10^8)$ evaluations are not uncommon).
As there are numerical evaluation routines available for multiple polylogarithms,
it is often beneficial -- whenever this is possible -- to express Feynman integrals and scattering amplitudes in terms of multiple polylogarithms.
This approach has its limitations when the expression in terms of multiple polylogarithms becomes too large and consequentially the numerical evaluation becomes too slow.
In this case one might prefer alternative methods, 
for example dedicated low-dimensional integral representations, which are evaluated numerically \cite{Chicherin:2017dob,Chicherin:2020oor,Chicherin:2021dyp}
or the numerical integration of the differential equation discussed below \cite{Liu:2022chg,Liu:2017jxz,Liu:2022mfb,Hidding:2020ytt}.

A popular technique to compute Feynman integrals analytically is the method of differential equations \cite{Kotikov:1990kg,Kotikov:1991pm,Remiddi:1997ny,Gehrmann:1999as}.
Within this approach, one first derives a system of differential equations for the yet unknown Feynman integrals.
In a second step one tries to find a transformation to an $\eps$-factorised form \cite{Henn:2013pwa}.
Finally, in a third step one solves the $\eps$-factorised system of differential equations order by order in $\eps$ in terms 
of iterated integrals.
The first and the third step are algorithmic.
Thus, the task of computing Feynman integrals reduces to the task of finding a transformation for a system of differential
equations to an $\eps$-factorised form.

Now let us assume that we have an $\eps$-factorised differential equation.
If all entries of the connection matrix $A$ are dlog-forms with arguments, which are rational functions
of the kinematic variables $x$, it is straightforward to show that all iterated integrals 
can be written as multiple polylogarithms.\footnote{For any path the arguments of the dlog-forms are rational functions of the path parameter. Finding (numerically) the roots of the polynomials in the numerator and in the denominator and by using 
$\ln(ab) = \ln a + \ln b$ and $\ln(\frac{a}{b}) = \ln a - \ln b$ one may write the iterated integrals in terms of multiple polylogarithms.}
Thus, all
Feynman integrals in this family can be expressed
in terms of multiple polylogarithms to all orders in the dimensional regularisation parameter $\eps$.

However, it is not uncommon that one has an $\eps$-factorised differential equation, where 
all entries of the connection matrix $A$ are dlog-forms with algebraic arguments and the
algebraic part is given by (multiple) square roots.
It is well-known that if all square roots are simultaneously rationalisable, 
the change of variables used in the rationalisation converts the dlog-forms with algebraic arguments
to dlog-forms with rational arguments and -- as a result --
the Feynman integrals can again be expressed
in terms of multiple polylogarithms.
This argument shows that the requirement that all square roots are simultaneously rationalisable is a sufficient condition
for expressing the Feynman integrals in terms of multiple polylogarithms.
However, it is not a necessary condition.
To see this, one may consider the example of the mixed QCD-electroweak corrections to the Drell--Yan process.
It is known that the three square roots appearing in this example cannot be rationalised simultaneously \cite{Besier:2019hqd}.
However, the result can be expressed in terms of multiple polylogarithms \cite{Heller:2019gkq}.
The latter can be shown using either direct integration \cite{Panzer:2014caa,Bourjaily:2021lnz} 
or symbol techniques \cite{Goncharov:2010jf,Duhr:2011zq}.
Hence, the requirement of being simultaneously rationalisable is not a necessary condition.
The simplified differential equation approach \cite{Papadopoulos:2014lla,Papadopoulos:2015jft,Syrrakos:2021nij}
provides a further technique by introducing an auxiliary parameter.

On the other side it is also known, that one may construct an example (outside the context of Feynman integrals) 
with an $\eps$-factorised differential equation with dlog-forms with algebraic arguments
where the resulting functions when integrated over a specific cycle
cannot be expressed in terms of multiple polylogarithms \cite{Duhr:2020gdd}.
Currently it is not known if this phenomenon can or cannot occur in the context of Feynman integrals.
 
In view of these findings it is highly desirable to find weaker conditions under which 
Feynman integrals obeying an $\eps$-factorised differential equation with dlog-forms with algebraic arguments can be expressed in term
of multiple polylogarithms.
This is the topic of this paper.
We discuss under which conditions we may use different rationalisations in different parts of the calculation.
To this aim we review path independence and parameterisation independence of iterated integrals.
Path independence holds only if a certain condition due to Chen is met. 
In the context of Feynman integrals we may divide the full result into smaller subsets,
consisting of specific linear combinations of iterated integrals, which are always path independent.
We may use different rationalisations in different subsets.

Furthermore, any iterated integral is invariant under re-parameterisation of the same integration path.
Hence, we may use within the same subset different rationalisations, if they correspond to different parameterisations of the same
integration path.
At first sight, this sounds very specific and it is not clear if there is actually a real application of this technique.
However, we present a non-trivial example -- the one-loop pentagon function with three adjacent massive external legs involving seven square roots --
where this technique can be used to express the result in terms of multiple polylogarithms.

In principle, we can push things even further:
Consider dividing a path-independent linear combination of iterated integrals into two subsets.
In general, these two subsets will not be individually path-independent.
However, one may add and subtract compensation terms to restore the path-independence.
This method is more general and has been outlined in ref.~\cite{Kreer:2021rim}, 
but comes at the price of introducing additional compensation terms.
The additional compensation terms may slow down the numerical evaluation.

We remark that the results of this article are also useful in the case, where all square roots are simultaneously rationalisable.
Rationalising all square roots simultaneously will introduce polynomials of high degree and in turn longer expressions in terms of multiple polylogarithms.
In this article we show under which conditions it is sufficient to rationalise only subsets of the square roots.
This will lead to polynomials of lower degree and in turn to shorter expressions in term of multiple polylogarithms.

This paper is organised as follows:
In the next section we introduce the general set-up and present the general theorems on 
path independence and parameterisation independence.
In section~\ref{sec:pentagon} we present the one-loop pentagon integral with three adjacent massive external legs.
In section~\ref{sect:conversion} we show explicitly how the techniques of section~\ref{sect:multiple_square_roots} are
applied to the one-loop pentagon integral. 
In section~\ref{sect:results} we discuss the results for one-loop pentagon integral with three adjacent massive external legs.
Finally, our conclusions are given in section~\ref{sec:conclusions}.
The article is supplemented by three appendices.
In appendix~\ref{sec:MomentumTwistorParametrization} we give details on
the momentum twistor representation,
in appendix~\ref{sec:R-termsInLetter} we list the polynomials entering the definition of the letters
and in appendix~\ref{sect:supplement} we describe the content of the
supplementary electronic file attached to the arxiv version of this article.

\section{Multiple square roots}
\label{sect:multiple_square_roots}

\subsection{Set-up}

We consider Feynman integrals, which depend on $\NB$ kinematic variables $\vec{x}=(x_1,\dots,x_{\NB})^T$.
We view the kinematic variables as coordinates on the kinematic space $X$.
Let $\vec{g}=(g_1,\dots,g_{\NF})^T$ be a vector of $\NF$ master integrals.
We assume that the master integrals satisfy an
$\eps$-factorised differential equation 
\bea
\label{diff_eq}
 d \vec{g}\left(\vec{x},\eps\right) & = & \eps \pmb{A}\left(\vec{x}\right) \vec{g}\left(\vec{x},\eps\right).
\eea
with an integrable connection $\pmb{A}$:
\bea
\label{integrability}
 \eps d \pmb{A} - \eps^2 \pmb{A} \wedge \pmb{A} & = & 0.
\eea
The differential $d$ is the differential in the kinematic variables $\vec{x}$:
\bea
 d & = & \sum\limits_{j=1}^{\NB} dx_j \frac{\partial}{\partial x_j}
\eea
As $\pmb{A}$ is independent of $\eps$, the $\eps^1$-term and the $\eps^2$-term in eq.~(\ref{integrability}) have to vanish separately and we 
get two individual equations
\bea
 d \pmb{A} \; = \; 0
 & \mbox{and} &
 \pmb{A} \wedge \pmb{A} \; = \; 0.
\eea
The first condition states that the entries of the $(\NF \times \NF)$-matrix $\pmb{A}$ are closed one-forms.
We denote a basis of the differential one-forms appearing in $\pmb{A}$ by $\omega_1,\dots,\omega_{\NL}$
and the ${\mathbb C}$-vector space spanned by those by $\Omega^1(X)$.
We write
\bea
 \pmb{A} & = & \sum\limits_{i=1}^{\NL} A_i \omega_i,
\eea
where the $A_i$'s are $(\NF \times \NF)$-matrices, whose entries are constant numbers.
We are in particular interested in the case, where the $\omega_i$'s are dlog-forms,
\bea
 \omega_i & = & d\ln\left(W_i\left(\vec{x}\right)\right),
\eea
where in turn the $W_i(\vec{x})$'s are algebraic functions of the kinematic variables $\vec{x}$.
The typical cases are that $W_i$ is either a rational function of $\vec{x}$ or of the form
\bea
\label{eq_algebraic_letter}
 W_i\left(\vec{x}\right) & = & \frac{P_i\left(\vec{x}\right)-\sqrt{Q_i\left(\vec{x}\right)}}{P_i\left(\vec{x}\right)+\sqrt{Q_i\left(\vec{x}\right)}},
\eea
where $P_i$ and $Q_i$ are polynomials in the variables $\vec{x}$.

We further assume that all master integrals have a Taylor expansion in the dimensional regularisation
parameter $\eps$, thus we may write
\bea
 \vec{g} \; = \; 
 \sum\limits_{j=0}^\infty \eps^j \vec{g}^{(j)}
 & \mbox{and} &
 g_k \; = \; 
 \sum\limits_{j=0}^\infty \eps^j g_k^{(j)},
 \;\;\; 1 \le k \le \NF.
\eea
Let $\vec{C}=(C_1,\dots,C_{\NF})^T$ be the boundary constants at $x=0$.
They have a similar Taylor expansion in $\eps$:
\bea
 C_k \; = \; 
 \sum\limits_{j=0}^\infty \eps^j C_k^{(j)},
 \;\;\; 1 \le k \le \NF.
\eea
Given the differential equation~(\ref{diff_eq}) and the boundary constants, we may express
the master integrals in terms of iterated integrals.
Let 
\bea
 \gamma & : & \left[0,1\right] \rightarrow X
\eea
be a parameterisation of a path in the kinematic space $X$ with starting point $\gamma(0)=\vec{0}$ and endpoint $\gamma(1)=\vec{x}$.
We write
\bea
 f_j\left(\lambda\right) d\lambda & = & \gamma^\ast \omega_j
\eea
for the pull-back of $\omega_j$  to the interval $[0,1]$.
If $f_r(\lambda)$ is regular at $\lambda=0$ we define the iterated integral by
\bea
 I_{\omega_1,\dots,\omega_r}\left[\gamma\right]
 & = &
 \int\limits_0^{\lambda} d\lambda_1 f_1\left(\lambda_1\right)
 \int\limits_0^{\lambda_1} d\lambda_2 f_2\left(\lambda_2\right)
 \dots
 \int\limits_0^{\lambda_{r-1}} d\lambda_r f_r\left(\lambda_r\right).
\eea
In case $f_r(\lambda)$ has a simple pole at $\lambda=0$ we use the standard ``trailing zero'' or tangential base point
prescription (see for example refs.~\cite{Brown:2013qva,Brown:2014aa,Walden:2020odh}).

Usually we evaluate an iterated integral along a standard path.
A typical example is given by the piecewise smooth path $\gamma_{\mathrm{standard}}$ consisting of a straight line from 
$(0,0,\dots,0,0)$ to $(0,0,\dots,0, \allowbreak x_{\NB})$, followed by a straight line from 
$(0,0,\dots,0,x_{\NB})$ to $(0,0,\dots,x_{\NB-1},x_{\NB})$, and the pattern continues in this way.
The last segment is given by straight line from $(0,x_2,\dots,x_{\NB-1},x_{\NB})$ to 
$(x_1, x_2, \dots, \allowbreak x_{\NB-1}, \allowbreak x_{\NB})$.

We may express the $\eps^r$-term of the $k$-th master integral
as a linear combination of iterated integrals of depth $\le r$
\bea
\label{eq_master_g_k_r}
 g_k^{(r)}\left(\vec{x}\right)
 & = &
 \sum\limits_{j=1}^r \sum\limits_{i_1,\dots,i_j=1}^{\NL}
 c^{(k, r)}_{i_1 \dots i_j} I_{\omega_{i_1},\dots,\omega_{i_j}}\left[\gamma\right].
\eea
We stress that an iterated integral $I_{\omega_{i_1},\dots,\omega_{i_r}}[\gamma]$
is a functional of the path $\gamma$ and not just a function of the endpoint $\vec{x}$.
We are interested in path independence under small variations of the path.
By a small variation of a path we mean a variation which does not change its homotopy class nor the trailing zero prescription.
In the following we will always use the term ``path independence'' as path independence under small variations.
The linear combination of iterated integrals appearing on the right-hand side 
of eq.~(\ref{eq_master_g_k_r}) is path-independent and defines therefore a function of the endpoint $\vec{x}$.
It is therefore worth reviewing under which conditions a linear combination of iterated integrals
is path-independent.

\subsection{Path independence}

Of particular interest are linear combinations of iterated integrals which for a 
fixed starting point $\vec{0}$ and fixed
end point $\vec{x}$ are independent of the path connecting $\vec{0}$ with $\vec{x}$.
A single iterated integral $I_{\omega_{i_1},\dots,\omega_{i_r}}\left[\gamma\right]$ 
is in general not path independent.
A criterium to decide if a given linear combination of iterated integrals is path independent can be stated
as follows \cite{Chen}:
There is a one-to-one correspondence between 
ordered sequences of differential one-forms $\omega_{i_1}$, $\omega_{i_2}$, $\dots$,$\omega_{i_r}$
and elements in the tensor algebra $(\Omega(X))^{\otimes r}$ 
(where $\Omega(X)$ denotes the vector space generated by the wedge products of $\omega_1, \dots, \omega_{\NL}$)
of the form
\bea
 \omega_{i_1} \otimes \omega_{i_2} \otimes \dots \otimes \omega_{i_r}.
\eea
It is customary to denote the latter as
\bea
 \left[ \omega_{i_1} | \omega_{i_2} | \dots | \omega_{i_r} \right]
 & = &
 \omega_{i_1} \otimes \omega_{i_2} \otimes \dots \otimes \omega_{i_r}.
\eea
In the tensor algebra we define (taking into account that all $\omega$'s are closed)
\bea
 d \left[ \omega_{i_1} | \omega_{i_2} | \dots | \omega_{i_r} \right]
 & = &
 \sum\limits_{j=1}^{r-1} \left[ \omega_{i_1} | \dots | \omega_{i_{j-1}} | \omega_{i_j} \wedge \omega_{i_{j+1}} | \omega_{i_{j+2}} | \dots | \omega_{i_r} \right].
\eea
Let us now consider a linear combination of iterated integrals of depth $\le r$ with constant coefficients
\bea
 g
 & = &
 \sum\limits_{j=1}^r \sum\limits_{i_1,\dots,i_j=1}^{\NL}
 c_{i_1 \dots i_j} I_{\omega_{i_1},\dots,\omega_{i_j}}\left[\gamma\right]
\eea
and the corresponding element in the tensor algebra
\bea
 B
 & = &
 \sum\limits_{j=1}^r \sum\limits_{i_1,\dots,i_j=1}^{\NL}
 c_{i_1 \dots i_j} \left[\omega_{i_1}|\dots|\omega_{i_j}\right].
\eea
$g$ is a homotopy functional (i.e. independent of small deformations of the integration path) if and only if 
\bea
\label{def_integrable_word}
 dB & = & 0.
\eea
The proof is due to Chen \cite{Chen}.
This is the sought-after criteria when a linear combination of iterated integrals is path independent.
We call any $B$ satisfying eq.~(\ref{def_integrable_word}) an 
integrable word.

While eq.~(\ref{def_integrable_word}) allows us to test easily if a given linear combination of 
iterated integrals is an integrable word, it does not lead in a simple way to a systematic procedure to construct
such a linear combination.
Nevertheless, there exist algorithms to find all integrable words.
For example, the program {\tt SymBuild} \cite{Mitev:2018kie} can be used to find all integrable words.

However, we know from the start, that the connection $\pmb{A}$ appearing in the differential equation
is integrable (or flat).
This guarantees that the linear combination of iterated integrals appearing in the solution 
for the $k$-th master integral $g_k$ is path independent.

We may refine this statement:
First of all, integrability holds for each order in the dimensional parameter $\eps$ independently, as we started 
from an $\eps$-factorised differential equation where the one-forms $\omega_i$ are independent of $\eps$.
This implies that the linear combination of iterated integrals appearing in the solution 
for the $\eps^r$-term of the $k$-th master integral $g^{(r)}_k(\vec{x})$ is path independent.

Secondly, there is a further refinement, which up to now and to the best of our knowledge
has not been stated clearly in the literature:
To prepare this statement let us first note
that each iterated integral appearing in $g^{(r)}_k(\vec{x})$ 
comes with a coefficient proportional to some boundary 
constant $C^{(j')}_{k'}$.
We may therefore consider the linear combination of iterated integrals appearing in $g^{(r)}_k(\vec{x})$
and being proportional to $C^{(j')}_{k'}$.
\begin{proposition}
\label{proposition_path_independence}
The linear combination of iterated integrals appearing in $g^{(r)}_k(\vec{x})$ and being proportional to $C^{(j')}_{k'}$
is by itself path independent.
\end{proposition}
\begin{proof}
The proof is rather simple: The iterated integrals appearing in $g^{(r)}_k(\vec{x})$ are path independent for any
value of the boundary constants, therefore the linear combination proportional to one particular boundary constant 
$C^{(j')}_{k'}$ must be path independent by itself.
\end{proof}
We illustrate this with a simple example:
Consider the $\eps$-factorised differential equation
\bea
 d 
 \left( \begin{array}{c}
 g_1\left(\vec{x},\eps\right) \\
 g_2\left(\vec{x},\eps\right) \\
 \end{array} \right)
 & = &
 \eps
 \left( \begin{array}{cc}
 \omega_{11} & 0 \\
 \omega_{21} & \omega_{22} \\
 \end{array} \right)
 \left( \begin{array}{c}
 g_1\left(\vec{x},\eps\right) \\
 g_2\left(\vec{x},\eps\right) \\
 \end{array} \right).
\eea
The connection matrix is assumed to be integrable, this translates to $\omega_{ij}$ being closed and
\bea
\label{example_integrability_condition}
 \omega_{21} \wedge \omega_{11} + \omega_{22} \wedge \omega_{21} & = & 0.
\eea
The general solution for $g_2(\vec{x},\eps)$ at order $\eps^2$ is
\bea
 g^{(2)}_2\left(\vec{x}\right)
 & = &
 C_2^{(2)}
 + C_1^{(1)} I_{\omega_{21}}\left[\gamma\right]
 + C_2^{(1)} I_{\omega_{22}}\left[\gamma\right]
 + C_1^{(0)} \left( I_{\omega_{21}, \omega_{11}}\left[\gamma\right] + I_{\omega_{22}, \omega_{21}}\left[\gamma\right] \right)
 + C_2^{(0)} I_{\omega_{22}, \omega_{22}}\left[\gamma\right].
\eea
The linear combination of iterated integrals proportional to each boundary constant $C_2^{(2)}, C_1^{(1)}, C_2^{(1)}, C_1^{(0)}$ and $C_2^{(0)}$ are individually path independent.
In particular, the linear combination of the iterated integrals proportional to $C_1^{(0)}$ 
\bea
 I_{\omega_{21}, \omega_{11}}\left[\gamma\right] + I_{\omega_{22}, \omega_{21}}\left[\gamma\right] 
\eea
is path independent. In this particular example, this can be verified directly 
from the integrability condition of eq.~(\ref{example_integrability_condition}).

Let us discuss typical applications of proposition~\ref{proposition_path_independence}: 
We say that a sub-sector is a daughter of a super-sector, if the sub-sector can be obtained from the super-sector by pinching one or more edges
of the graph of the super-sector.
\begin{figure}
\begin{center}
\includegraphics[scale=1.0]{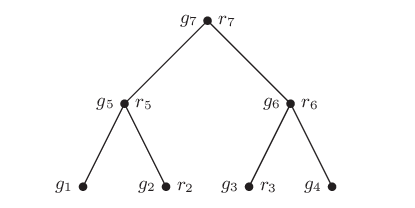}
\includegraphics[scale=1.0]{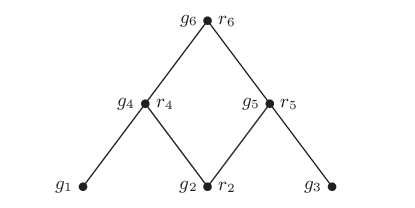}
\end{center}
\caption{
Left: A system with $7$ master integrals, where the sub-sectors corresponding to $g_5$ and $g_6$ do not have a common daughter.
Right: A system with $6$ master integrals, where the sub-sectors corresponding to $g_4$ and $g_5$ have a common daughter $g_2$.
}
\label{fig_sector_tree}
\end{figure}
Consider now a system of Feynman integrals, where the nesting of the sub-sectors is as indicated 
in the left part of fig.~\ref{fig_sector_tree}
and suppose further that some sectors involve square roots, again as indicated 
in the left part of fig.~\ref{fig_sector_tree}.
Let us assume that we are able to rationalise simultaneously the square roots $(r_2,r_5,r_7)$
as well as $(r_3,r_6,r_7)$.
However, let us assume that it is impossible to rationalise simultaneously all square roots $(r_2,r_3,r_5,r_6,r_7)$.
In this case, proposition~\ref{proposition_path_independence} can be applied: There will be no iterated integrals involving all five square roots.
Each iterated integral will involve at most three square roots, either the triple $(r_2,r_5,r_7)$ or the triple $(r_3,r_6,r_7)$.
The iterated integrals involving the triple $(r_2,r_5,r_7)$ will necessarily be proportional to $C_2^{(j)}$,
while the iterated integrals involving the triple $(r_3,r_6,r_7)$ will be proportional to $C_3^{(j)}$.
Proposition~\ref{proposition_path_independence} states that we may choose different integration paths and 
rationalise these triples of square roots independently.
Iterated integrals proportional to the other boundary constants will involve at most two square roots, either $(r_5,r_7)$ or $(r_6,r_7)$.
For the pair $(r_5,r_7)$ we may use the same rationalisation as for $(r_2,r_5,r_7)$, 
similarly we may use for the pair $(r_6,r_7)$
the same rationalisation as for $(r_3,r_6,r_7)$.
We remark that it can be advantageous to use simpler
rationalisation for the pairs of square roots.
For example, we may use for the iterated integrals proportional to $C_1^{(j)}$ or $C_5^{(j)}$ a rationalisation, which just rationalises $r_5$ and $r_7$, but not $r_2$.
This may lead to more compact expressions for these terms.

We remark that proposition~\ref{proposition_path_independence} can be useful even in the case, where all square roots are simultaneously rationalisable.
Rationalising all square roots simultaneously will introduce polynomials of high degree and in turn longer expressions in terms of multiple polylogarithms.
The terms proportional to a particular boundary constant may involve only a subset of the square roots, and rationalising just these square roots
may result in shorter expressions.

Quite often it is the case that different sub-sectors have a common daughter. An example
is shown in the right part of fig.~\ref{fig_sector_tree}.
Let us assume
that we are able to rationalise simultaneously the square roots $(r_2,r_4,r_6)$
as well as $(r_2,r_5,r_6)$, but not $(r_2,r_4,r_5,r_6)$.
Proposition~\ref{proposition_path_independence} can be applied to all terms, which are not proportional to the boundary constants $C_2^{(j)}$.
The iterated integrals proportional to $C_2^{(j)}$ will involve all four square roots.
Any iterated integral in this linear combination will involve either the triple $(r_2,r_4,r_6)$ or the triple $(r_2,r_5,r_6)$, but never all four of them.
However, the linear combination can in general not be split into two subsets 
which are individually path-independent.
Insisting in splitting the linear combination into two path independent subsets will in general require the addition and subtraction 
of compensation terms, see ref.~\cite{Kreer:2021rim}.
However, there can be an alternative: Two different parameterisations of the same integration path may rationalise either 
$(r_2,r_4,r_6)$ or $(r_2,r_5,r_6)$.
This possibility will be discussed in the next subsection.

\subsection{Parameterisation independence}

In this subsection we consider different parameterisations of the same path (we speak about the same path if
the image in $X$ is the same).
To this aim let
\bea
 \gamma_1 & : & \left[0,1\right] \rightarrow X
 \nonumber \\
 & & \lambda_1 \rightarrow \gamma_1\left(\lambda_1\right) 
\eea
and
\bea
 \gamma_2 & : & \left[0,1\right] \rightarrow X
 \nonumber \\
 & & \lambda_2 \rightarrow \gamma_2\left(\lambda_2\right) 
\eea
be two different parameterisations of the same path, i.e.
\bea
 \mathrm{Im}\left(\gamma_1\right)
 & = &
 \mathrm{Im}\left(\gamma_2\right)
\eea
We may express $\lambda_2$ in terms of $\lambda_1$ and vice versa:
\bea
 \lambda_2 \; = \; \gamma_2^{-1}\left(\gamma_1\left(\lambda_1\right)\right),
 & &
 \lambda_1 \; = \; \gamma_1^{-1}\left(\gamma_2\left(\lambda_2\right)\right).
\eea
We require in addition that
\bea
\label{eq_tangential_base_point}
 \left. \frac{d\lambda_2}{d\lambda_1} \right|_{\lambda_1=0} & = & 1.
\eea
An iterated integral is independent under these re-parameterisations of the path:
\bea
 I_{\omega_1,\dots,\omega_r}\left[\gamma_1\right]
 & = &
 I_{\omega_1,\dots,\omega_r}\left[\gamma_2\right].
\eea
The additional condition in eq.~(\ref{eq_tangential_base_point}) is required for iterated integrals with trailing zeros.

\subsection{The treatment of multiple square roots}

With these prerequisites at hand, we may now turn to our main problem: The treatment of multiple square
roots in the arguments of the dlog-forms.
We consider a situation where all differential one-forms $\omega_i$ are dlog-forms, but
some arguments of the logarithms are algebraic functions of the kinematic variables involving square
roots.

It is well-known that if all occurring square roots are simultaneously rationalisable, then the result
can be expressed in terms of multiple polylogarithms.
Being simultaneously rationalisable is a sufficient condition, but not a necessary condition.
As we advance in Feynman integral calculations we encounter more and more examples with multiple square
roots which cannot be rationalised simultaneously.
This does not necessarily mean that the Feynman integrals cannot be epxressed in terms of 
multiple polylogarithms.
Alternative methods like direct integration or symbol calculus might lead to a result in terms of 
multiple polylogarithms \cite{Heller:2019gkq,Bonetti:2020hqh}.
As the method of differential equations is the most popular method for computing Feynman integrals,
we will give below a weaker sufficient condition under which the master integrals can be expressed 
(to any order in the dimensional parameter $\eps$) in terms of multiple polylogarithms.

The first observation is that often it is impossible that all square roots appear 
in a single iterated integral.
A typical example is the situation, where a square root $r_1$ is associated with one particular sub-sector,
a second square root $r_2$ with a second sub-sector with the same number of propagators
and the two square roots appear only in these two sub-sectors.
Then we will never have an iterated integral (to any order in the dimensional parameter $\eps$)
involving both square roots $r_1$ and $r_2$, as the differential equation does not couple 
sub-sectors with the same number of propagators.

We are therefore tempted to consider different rationalisations for the subsets of square roots which do occur 
in a single iterated integral.
Within a given rationalisation we would like to use again a simple integration path (like $\gamma_{\mathrm{standard}}$) for the integration.
The complication we have to face is the following: A simple integration path in the new variables
does not necessarily correspond to the integration path in the original variables.
If we are going to use different integration paths for different iterated integrals, the question
of path-independence becomes relevant.

With our previous results, we may now state a weaker sufficient condition under which we may express
the master integrals in terms of multiple polylogarithms:
We are free to use different integration paths (and hence different rationalisations) for any 
linear combination of iterated integrals, which are by themselves path-independent.
For example, we may collect in the $\eps^r$-term of the $k$-th master integral the iterated integrals
proportional to the boundary constant $C_{k'}^{(j')}$.
For each boundary constant we may use a different integration path and therefore a different 
rationalisation.

Within each path-independent linear combination of iterated integrals we have to use the same
integration path, but we may use different parameterisations of the same integration path.
Different parameterisations of the same integration path may correspond to different rationalisations.
A typical example is the following:
Consider the integration path $\gamma_{\mathrm{standard}}$ and a transformation
\bea
\label{eq_trafo_x1}
 x_1 & = & f\left(t_1,x_2,\dots,x_{\NB}\right),
\eea
such that $x_1=0$ corresponds to $t_1=0$ and
\bea
 \left. \frac{dx_1}{dt_1} \right|_{t_1=0} & = & 1.
\eea
This will replace the variable $x_1$ by $t_1$. The transformation in eq.~(\ref{eq_trafo_x1}) may depend
on the additional variables $x_2,\dots,x_{\NB}$.
However, it does not change the integration path $\gamma_{\mathrm{standard}}$: 
The path lying in the hyperplane $x_1=0$ is not affected at all and for the last segment from
$(0,x_2,\dots,x_{\NB-1},x_{\NB})$ to 
$(x_1, x_2, \dots, x_{\NB-1}, x_{\NB})$
it is just a re-parameterisation, as $x_2, \dots, x_{\NB}$ are constant along this segment.

In the following we present a highly non-trivial example where this happens in an actual calculation:
The one-loop pentagon integral with three adjacent massive external legs
is an example which involves seven square roots.\footnote{The one-loop pentagon integral with three non-adjacent legs is simpler, as it involves five square roots, and has already been computed with the simplified differential equations approach~\cite{Syrrakos:2021nij}.}
We are not able to rationalise simultaneously all seven square roots.\footnote{We did not attempt to strictly prove that it is impossible to rationalise the seven square roots simultaneously. 
Such proofs -- if desired -- can be constructed along the lines of refs.~\cite{Besier:2020hjf,Festi:2021tyq}.}
On combinatorial grounds only three square roots will ever appear in a single iterated integral.
We show that any triple of occuring square roots can be rationalised by a suitable re-parameterisation
of the standard path $\gamma_{\mathrm{standard}}$. 
The re-parameterisations are of the form as in eq.~(\ref{eq_trafo_x1}).
This establishes that the one-loop three-mass hard pentagon integral can be expressed
to all orders in the dimensional regularisation parameter $\eps$
in terms of multiple polylogarithms.

\section{The one-loop pentagon integral with three adjacent massive external legs}
\label{sec:pentagon}

\subsection{Definition of the Feynman integral}

We consider the family of Feynman integrals shown in \figref{Fig:MassivePentagon} and given by 
\bea
 G\left[\nu_1,\nu_2,\nu_3,\nu_4,\nu_5\right] 
 =
 e^{\epsilon\gamma_E}(\mu^2)^{\nu-\frac{D}{2}}
 \int\,\frac{d^{D}l}{i\pi^{D/2}}\,\frac{1}{D_1^{\nu_1}\,D_2^{\nu_2}\,D_3^{\nu_3}\,D_4^{\nu_4}\,D_5^{\nu_5}}, \Label{Eq:PentagonIntegral}
\eea
where $l$ is the loop momentum and $D$ denotes its dimension.
Unless stated otherwise we take $D=4-2\epsilon$.  
The arbitrary scale $\mu$ is introduced to render the integral dimensionless.
Three consecutive external legs  are massive, $p_i^2=m_i^2$ with $i=1,2,3$, whereas $p_4^2=p_5^2=0$. The inverse propagators $D_i$ are
\begin{figure}
	\centering
	\includegraphics[width=0.3\textwidth]{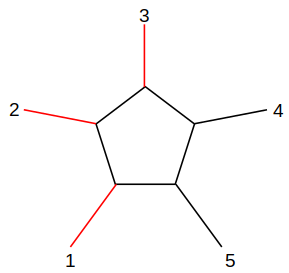}
	\caption{Pentagon topology with three adjacent massive external legs denoted in red. All other external legs and internal propagators are massless.}
	\label{Fig:MassivePentagon}
\end{figure}
\bea
&&D_1=l^2,\qquad D_2=(l-p_1)^2,\qquad D_3=(l-p_1-p_2)^2, \nn
&&D_4=(l-p_1-p_2-p_3)^2,\qquad D_5=(l-p_1-p_2-p_3-p_4)^2,
\eea
and there are eight (dimensionful) kinematic variables for this integral,
\bea
\vec{v}&=&\{s_{12},s_{23},s_{34},s_{45},s_{15},m_1^2,m_2^2,m_3^2\}, \Label{Eq:KinematicVariables1}
\eea
where $s_{i,i+1}=(p_i+p_{i+1})^2$ denotes the Mandelstam variables.
As usual, we may set $\mu^2$ without loss of generality to a specific value.
Doing so, the integral will depend on seven dimensionless kinematic variables.
For example, setting  $\mu^2=-m_1^2$, the integral will depend on the seven dimensionless kinematic variables
\bea
 \{\frac{s_{12}}{m_1^2},\frac{s_{23}}{m_1^2},\frac{s_{34}}{m_1^2},\frac{s_{45}}{m_1^2},\frac{s_{15}}{m_1^2},\frac{m_2^2}{m_1^2},\frac{m_3^2}{m_1^2}\}.
\eea
However, in this concrete example it is slightly more convenient to keep 
the dependence on the eight dimensionful kinematic variables, as we will use a twistor parameterisation for the latter.

\subsection{Master integrals}

\begin{figure}
	\centering
	\includegraphics[width=1\textwidth]{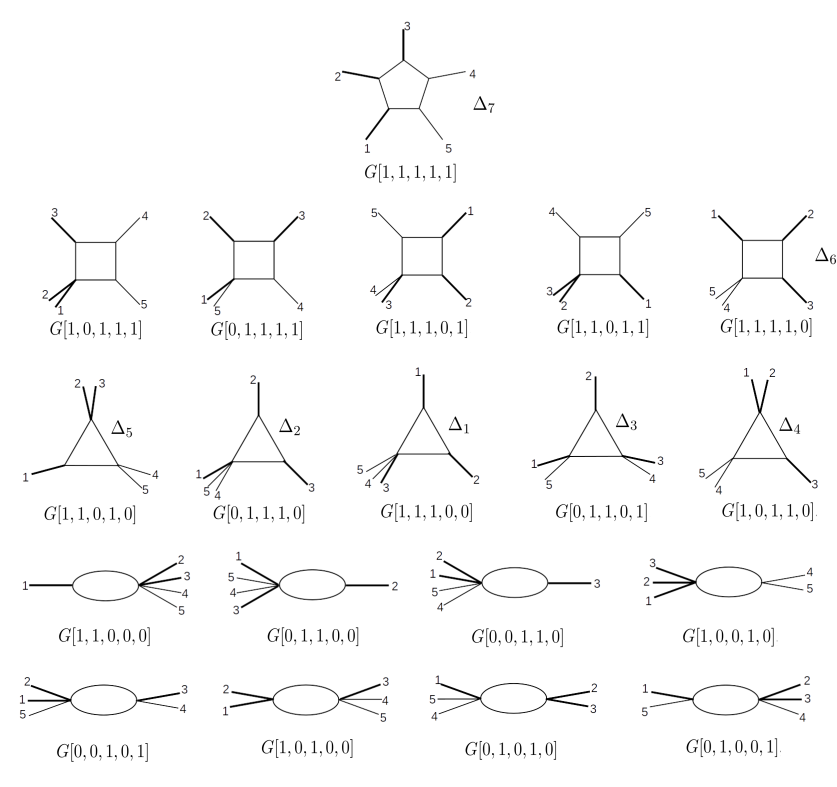}
	\caption{All master integrals of the one-loop three-mass hard pentagon topology, with thick lines denoting the massive external legs. The leading singularities of a subset thereof will depend on the square roots~(\ref{eq:squareroots}), as indicated in the figure.}
	\label{Fig:MasterIntegral}
\end{figure}

All of the Feynman integrals in the pentagon integral family $G[\nu_1,\ldots,\nu_5]$ can be reduced to a set of master integrals with the help of integration-by-parts relations. With the help of e.g. \texttt{LiteRed} \cite{Lee:2013mka}, these may be solved in terms of the 19 master integrals shown in Figure~\ref{Fig:MasterIntegral}.
As zero exponents are equivalent to removing the corresponding propagator or in other words contracting the corresponding edge, it is evident that this set of master integrals includes 8 bubble integrals, 5 triangle integrals, 5 box integrals, and 1 pentagon integral. 

The master integrals of Figure~\ref{Fig:MasterIntegral} are not of uniform transcendental (UT) weight
and do not lead to an $\eps$-factorised differential equation.
In the following we will present a basis of master integrals, which are of uniform transcendental weight.
In the case of one-loop integrals this is straightforward.
Integrals of uniform transcendental weight are expected to have constant leading singularities~\cite{Arkani-Hamed:2010pyv}.
We choose the integer dimension of the loop momentum of $n$-point 1-loop integrals to be $2\lfloor\frac{n+1}{2}\rfloor$, 
where $\lfloor x \rfloor$ denotes the integer part of $x$,
and divide these integrals by their leading singularities \cite{Abreu:2017mtm}. 
Dimensional recurrence relations~\cite{Tarasov:1996br,Lee:2009dh} relate integrals with $D$ and $D-2$ dimensions, thus if desired these can be used to re-express our basis in terms of integrals with only four integer dimensions.

We start with the bubble integrals.
It is easy to show that the bubble integral $G[0,0,1,0,1]$ has leading singularity $\frac{1}{s_{34}}$. 
Thus the integral $s_{34}G[0,0,1,0,1]$ is uniformly transcendental in $D=2-2\epsilon$, and with the help of dimensional recurrence relations we can trade this with the uniformly transcendental integral $s_{34}\, G[0,0,2,0,1]$ in $D=4-2\epsilon$. 
In the same fashion, the eight bubble integrals in our UT basis will be, 
\bea
g_1&=&\epsilon\,\frac{s_{34}}{\mu^2}\, G[0,0,2,0,1],\qquad
 g_2=\epsilon\,\frac{m_3^2}{\mu^2}\, G[0,0,2,1,0],\nn
g_3&=&\epsilon\,\frac{s_{15}}{\mu^2}\, G[0,2,0,0,1],\qquad 
g_4=\epsilon\,\,\frac{s_{23}}{\mu^2}\, G[0,2,0,1,0],\nn
g_5&=&\epsilon\,\frac{m_2^2}{\mu^2} \, G[0,2,1,0,0],\qquad 
g_6=\epsilon\,\frac{s_{45}}{\mu^2}\, G[2,0,0,1,0],\nn
g_7&=&\epsilon\,\frac{s_{12}}{\mu^2}\, G[2,0,1,0,0],\qquad 
g_8=\epsilon\,\frac{m_1^2}{\mu^2}\, G[2,1,0,0,0],
\eea
where the pre-factor $\epsilon$ ensures that the integrals are of uniform transcendental weight zero (if we count the dimensionless parameter with weight $-1$).

The external legs of each triangle integral are all massive, and the square roots shown in the first five lines of ~\eqref{eq:squareroots} below appear in the denominators of their leading singularities. 
The triangle integrals in our basis are
\bea
g_9&=&\epsilon^2\,\frac{\Delta_3}{2\mu^2}\,G[0,1,1,0,1], \qquad
g_{10}=\epsilon^2\,\frac{\Delta_2}{2\mu^2}\,G[0,1,1,1,0], \nn
g_{11}&=&\epsilon^2\,\frac{\Delta_4}{2\mu^2}\,G[1,0,1,1,0],\qquad 
g_{12}=\epsilon^2\,\frac{\Delta_5}{2\mu^2}\,G[1,1,0,1,0], \nn 
g_{13}&=&\epsilon^2\,\frac{\Delta_1}{2\mu^2}\,G[1,1,1,0,0].
\eea
Similarly, the box integrals in our basis are
\bea
g_{14}&=&\epsilon^2\,\frac{s_{23}s_{34}-m_3^2s_{15}}{2\mu^4}\,G[0,1,1,1,1], \nn 
g_{15}&=&\epsilon^2\,\frac{s_{34}s_{45}}{2\mu^4}G[1,0,1,1,1], \nn
g_{16}&=&\epsilon^2\,\frac{s_{45}s_{15}}{2\mu^4}\,G[1,1,0,1,1],\nn
g_{17}&=&\epsilon^2\,\frac{s_{15}s_{12}-m_1^2s_{34}}{2\mu^4}\,G[1,1,1,0,1], \nn 
g_{18}&=&\epsilon^2\,\frac{\Delta_6}{2\mu^4}\,G[1,1,1,1,0],\nn
\eea
where $g_{15}$ and $g_{16}$ have two massive external legs, $g_{14}$ and $g_{17}$ have three massive external legs, and $g_{18}$ has all external legs massive, with $\Delta_6$ also given in~\eqref{eq:squareroots}.

Finally, the pentagon integral has leading singularity $\frac{1}{\Delta_7}$ in $D=6-2\epsilon$, and the corresponding UT integral is
\bea
g_{19}&=&\epsilon^3\,\frac{\Delta_7}{\mu^4} G^{D=6-2\epsilon}[1,1,1,1,1],
\eea
noting in particular the integer dimension choice indicated at the beginning of this subsection. With the help of dimensional recurrence relations~\cite{Tarasov:1996br,Lee:2009dh} we can always write this integral as a linear combination of integrals in $D=4-2\epsilon$, hence this merely corresponds to a convenient notation.

Finally, the square roots appearing in the above formulas are given by 
\bea\label{eq:squareroots}
\Delta_1&=&\sqrt{m_1^4-2m_1^2m_2^2+m_2^4-2m_1^2s_{12}-2m_2^2s_{12}+s_{12}^2}, \nn
\Delta_2&=& \sqrt{m_2^4-2m_2^2m_3^2+m_3^4-2m_2^2s_{23}-2m_3^2 s_{23}+s_{23}^2}, \nn
\Delta_3&=&\sqrt{m_2^4-2m_2^2 s_{15} + s_{15}^2-2m_2^2 s_{34}-2s_{15}s_{34}+s_{34}^2}, \nn
\Delta_4&=&\sqrt{m_3^4-2m_3^2 s_{12}+s_{12}^2-2m_3^2 s_{45}-2s_{12} s_{45}+s_{45}^2}, \nn
\Delta_5&=&\sqrt{m_1^4-2m_1^2 s_{23}+s_{23}^2-2m_1^2 s_{45}-2s_{23}s_{45}+s_{45}^2}, \nn 
\Delta_6&=&\sqrt{m_1^4m_3^4-2m_1^2 m_3^2 s_{12}s_{23}+s_{12}^2 s_{23}^2-2m_1^2 m_2^2m_3^2s_{45}-2m_2^2 s_{12}s_{23}s_{45}+m_2^4 s_{45}^2},  \nn
\Delta_7&=&\sqrt{\Lambda(p_1,p_2,p_3,p_4)}=\sqrt{\det(-p_i\cdot p_j)}.
\eea
which unfortunately cannot be rationalized simultaneously \cite{Besier:2020hjf,Weinzierl:2022eaz}.

\subsection{The alphabet}
\label{sec:alphabet}

The differential equation for the basis $\vec{g}=\{g_1,...,g_{19}\}$ has the $\eps$-factorised form
\bea
d\,\vec{g}=\epsilon\,\pmb{A}(\vec{v})\, \vec{g}\,,\quad \pmb{A}(\vec{v})=\Sl_{i=1}^{57}\,A_i\,d\log\left(\frac{W_i}{\left(\mu^2\right)^{\alpha_i}}\right)\,, \Label{Eq:DiffEquation}
\eea
where the matrix $\pmb{A}$ is independent of $\epsilon$, and its dependence on $\vec{v}$ comes only through the \emph{letters} $W_i$, whereas their coefficients $A_i$ are constant matrices. 
The exponent $\alpha_i$ is chosen such that $\frac{W_i}{\left(\mu^2\right)^{\alpha_i}}$ is dimensionless.
To make contact with the previous section we set
\bea
 \omega_i & = & d\log\left(\frac{W_i}{\left(\mu^2\right)^{\alpha_i}}\right).
\eea
For this integral family we have 57 letters in the differential equation. We will provide their explicit form in a moment, but together with $\pmb{A}$ these may also be found in the ancillary file attached to the arxiv version of this article,
see appendix~\ref{sect:supplement} for details.
We set
\bea
 {\mathcal A}
 & = &
 \left\{ \omega_1, \dots, \omega_{57}\right\}.
\eea
Expressed in terms of the kinematic variables $v_i$, the total differential in the above equation takes the form
\bea
d=\Sl_{i=1}^{8}\,dv_i\,\frac{\partial}{\partial v_i}.
\eea
The 57 letters $\{W_i\}$ appearing in the matrix $\pmb{A}$ of the differential equation (\ref{Eq:DiffEquation}) may be obtained from the general results for one-loop alphabets worked out in~\cite{Abreu:2017mtm, Chen:2022fyw,Dlapa:2023cvx}. In particular, by taking the limit of the generic to the three-mass hard pentagon integral with the help of the ancillary file of~\cite{Dlapa:2023cvx}, and eliminating any multiplicative dependence, we arrive at the following alphabet: 
\bea
W_1&=&-m_1^2,\qquad~~~  W_2=-m_2^2,\qquad~~~ 
W_3=-m_3^2,\qquad~~~ W_4=-s_{12},\nn
W_5&=&m_1^4-2 m_2^2 m_1^2-2 s_{12} m_1^2+m_2^4+s_{12}^2-2 m_2^2 s_{12},\qquad~
W_6=m_1^2-s_{15},\nn
W_7&=&-s_{15},\qquad~~~
W_8=s_{15}-s_{23},\quad~
W_9=-s_{23},\qquad~~~ 
W_{11}=m_3^2-s_{34},\nn
W_{10}&=&m_2^4-2 m_3^2 m_2^2-2 s_{23} m_2^2+m_3^4+s_{23}^2-2 m_3^2 s_{23},\qquad~
W_{12}=s_{12}-s_{34},\nn
W_{13}&=&s_{34},\qquad~~~~~
W_{14}=s_{12} s_{15}-m_1^2 s_{34},\qquad\qquad
W_{15}=m_3^2 s_{15}-s_{23} s_{34},\nn
W_{16}&=&m_2^4-2 s_{15} m_2^2-2 s_{34} m_2^2+s_{15}^2+s_{34}^2-2 s_{15} s_{34}, \nn
W_{17}&=&-s_{34} m_1^4-s_{34}^2 m_1^2+s_{12} s_{15} m_1^2+m_2^2 s_{34} m_1^2+s_{12} s_{34} m_1^2+s_{15} s_{34} m_1^2-s_{12} s_{15}^2 \nn
&&~~~-m_1^2 m_2^2 s_{12}-s_{12}^2 s_{15}+m_2^2 s_{12} s_{15}-m_2^2 s_{15} s_{34}+s_{12} s_{15} s_{34}, \nn
W_{18}&=&-s_{15} m_3^4-s_{15}^2 m_3^2+m_2^2 s_{15} m_3^2-m_2^2 s_{23} m_3^2+s_{15} s_{23} m_3^2+s_{15} s_{34} m_3^2+s_{23} s_{34} m_3^2 \nn
&&~~~-s_{23} s_{34}^2-s_{23}^2 s_{34}-m_2^2 s_{15} s_{34}+m_2^2 s_{23} s_{34}+s_{15} s_{23} s_{34}, \nn
W_{20}&=&s_{15}^2-m_1^2 s_{15}-s_{23} s_{15}+s_{45} s_{15}+m_1^2 s_{23},\qquad
W_{19}=-s_{45},\nn
W_{21}&=&s_{34}^2-m_3^2 s_{34}-s_{12} s_{34}+s_{45} s_{34}+m_3^2 s_{12}, \nn
W_{22}&=&m_3^4-2 s_{12} m_3^2-2 s_{45} m_3^2+s_{12}^2+s_{45}^2-2 s_{12} s_{45},\nn
W_{23}&=&m_1^4-2 s_{23} m_1^2-2 s_{45} m_1^2+s_{23}^2+s_{45}^2-2 s_{23} s_{45},\nn
W_{24}&=&\frac{m_1^2+m_2^2-s_{12}-\Delta_1}{m_1^2+m_2^2-s_{12}+\Delta_1},\qquad\qquad\quad
W_{25}=\frac{m_1^2-m_2^2+s_{12}-\Delta_1}{m_1^2-m_2^2+s_{12}+\Delta_1},\nn
W_{26}&=&\frac{m_2^2+m_3^2-s_{23}-\Delta_2}{m_2^2+m_3^2-s_{23}+\Delta_2},\qquad\qquad\quad
W_{27}=\frac{m_2^2-m_3^2+s_{23}-\Delta_2}{m_2^2-m_3^2+s_{23}+\Delta_2},\nn
W_{28}&=&\frac{R_{28}-s_{34} m_1^2 \Delta_1+s_{12} s_{15} \Delta_1}{R_{28}+s_{34} m_1^2 \Delta_1-s_{12} s_{15} \Delta_1},\qquad
W_{29}=\frac{R_{29}+s_{15} m_3^2\Delta_2-s_{23} s_{34} \Delta_2}{R_{29}-s_{15} m_3^2\Delta_2+s_{23} s_{34} \Delta_2},\nn
W_{30}&=&\frac{m_2^2+s_{15}-s_{34}-\Delta_3}{m_2^2+s_{15}-s_{34}+\Delta_3},\qquad\qquad~~~
W_{31}=\frac{m_2^2-s_{15}+s_{34}-\Delta_3}{m_2^2-s_{15}+s_{34}+\Delta_3},\nn
W_{32}&=&\frac{R_{32}+s_{12} \Delta_3 s_{15}-m_1^2 s_{34} \Delta_3}{R_{32}-s_{12} \Delta_3 s_{15}+m_1^2 s_{34} \Delta_3},\quad~~~
W_{33}=\frac{R_{33}+m_3^2 \Delta_3 s_{15}-s_{23} s_{34} \Delta_3}{R_{33}-m_3^2 \Delta_3 s_{15}+s_{23} s_{34} \Delta_3},\nn
W_{34}&=&\frac{m_3^2 s_{12} s_{15}^2-m_1^2 m_3^2 s_{34} s_{15}-s_{12} s_{23} s_{34} s_{15}+m_2^2 s_{34} s_{45} s_{15}+m_1^2 s_{23} s_{34}^2}{R_{34}}, \nn
W_{35}&=&\frac{R_{35}}{m_1^4 m_3^4-2 m_1^2 s_{12} s_{23} m_3^2-2 m_1^2 m_2^2 s_{45} m_3^2+s_{12}^2 s_{23}^2+m_2^4 s_{45}^2-2 m_2^2 s_{12} s_{23} s_{45}},\nn
W_{36}&=&\frac{m_3^2+s_{12}-s_{45}-\Delta_4}{m_3^2+s_{12}-s_{45}+\Delta_4},\qquad\quad~~~~~
W_{37}=\frac{m_3^2-s_{12}+s_{45}-\Delta_4}{m_3^2-s_{12}+s_{45}+\Delta_4},\nn
W_{38}&=&\frac{m_1^2+s_{23}-s_{45}-\Delta_5}{m_1^2+s_{23}-s_{45}+\Delta_5},\qquad\quad~~~~~
W_{39}=\frac{m_1^2-s_{23}+s_{45}-\Delta_5}{m_1^2-s_{23}+s_{45}+\Delta_5},\nn
W_{40}&=&\frac{-2 m_3^2 s_{12}+s_{34} s_{12}+m_3^2 s_{34}-s_{34} s_{45}-s_{34} \Delta_4}{-2 m_3^2 s_{12}+s_{34} s_{12}+m_3^2 s_{34}-s_{34} s_{45}+s_{34} \Delta_4}, \nn
W_{41}&=&\frac{m_1^2 s_{15}+s_{23} s_{15}-s_{45} s_{15}-\Delta_5 s_{15}-2 m_1^2 s_{23}}{m_1^2 s_{15}+s_{23} s_{15}-s_{45} s_{15}+\Delta_5 s_{15}-2 m_1^2 s_{23}},\nn
W_{42}&=&\frac{m_1^2 m_3^2+s_{12} s_{23}-m_2^2 s_{45}-\Delta_6}{m_1^2 m_3^2+s_{12} s_{23}-m_2^2 s_{45}+\Delta_6},\qquad
W_{43}=\frac{m_1^2 m_3^2-s_{12} s_{23}+m_2^2 s_{45}-\Delta_6}{m_1^2 m_3^2-s_{12} s_{23}+m_2^2 s_{45}+\Delta_6},\nn
W_{44}&=&\frac{R_{44}-\Delta_1 \Delta_6}{R_{44}+\Delta_1 \Delta_6},\qquad
W_{45}=\frac{R_{45}-\Delta_2 \Delta_6}{R_{45}+\Delta_2 \Delta_6},\qquad~
W_{46}=\frac{R_{46}-\Delta_4 \Delta_6}{R_{46}+\Delta_4 \Delta_6},\nn
W_{47}&=&\frac{R_{47}-\Delta_5 \Delta_6}{R_{47}+\Delta_5 \Delta_6},\qquad
W_{48}=\frac{R_{48}+\Delta_7}{-R_{48}+\Delta_7},\qquad~~
W_{49}=\frac{R_{49}-\Delta_1 \Delta_7}{R_{49}+\Delta_1 \Delta_7},\nn
W_{50}&=&\frac{R_{50}-\Delta_7 s_{15}}{R_{50}+\Delta_7 s_{15}},\qquad
W_{53}=\frac{R_{53}-\Delta_2 \Delta_7}{R_{53}+\Delta_2 \Delta_7},\qquad
W_{54}=\frac{R_{54}-\Delta_3 \Delta_7}{R_{54}+\Delta_3 \Delta_7},\nn
W_{51}&=&\frac{R_{51}+\Delta_7 m_1^2-s_{15} \Delta_7}{R_{51}-\Delta_7 m_1^2+s_{15} \Delta_7},\qquad
W_{52}=\frac{R_{52}+\Delta_7 s_{15}-s_{23} \Delta_7}{R_{52}-\Delta_7 s_{15}+s_{23} \Delta_7},\nn
W_{55}&=&\frac{R_{55}-\Delta_4 \Delta_7}{R_{55}+\Delta_4 \Delta_7},\qquad
W_{56}=\frac{R_{56}-\Delta_5 \Delta_7}{R_{56}+\Delta_5 \Delta_7},\qquad
W_{57}=\frac{R_{57}-\Delta_6 \Delta_7}{R_{57}+\Delta_6 \Delta_7}, \Label{Eq:PentagonLetter}
\eea
where the $R_i$ denote large polynomials in the kinematic variables, whose precise form has been relegated to  \appref{sec:R-termsInLetter}. 

The seven roots $\Delta_1-\Delta_7$ are associated to specific sub-sectors, in particular the roots $\Delta_1-\Delta_5$
are associated to the triangle diagrams.
From the block triangular structure of the matrix $\pmb{A}$ it follows, that the roots $\Delta_1-\Delta_5$ can never occur simultaneously 
in any given iterated integral.
Furthermore, as the triangle $G[0,1,1,0,1]$ is not a sub-sector of the box integral $G[1,1,1,1,0]$ it follows that the roots
$\Delta_3$ and $\Delta_6$ can never 
occur simultaneously 
in any given iterated integral.
Therefore, the only sub-sets of roots which will appear in any given iterated integral are
\bea
&&(\Delta_1,\Delta_6,\Delta_7),\quad (\Delta_2,\Delta_6,\Delta_7),\quad
(\Delta_4,\Delta_6,\Delta_7),\quad (\Delta_5,\Delta_6,\Delta_7), \nn
&&(\Delta_1,\Delta_7),\quad (\Delta_2,\Delta_7),\quad (\Delta_3,\Delta_7),\quad (\Delta_4,\Delta_7), \quad (\Delta_5,\Delta_7),\quad (\Delta_6,\Delta_7), \nn
&&
 (\Delta_1,\Delta_6),\quad (\Delta_2,\Delta_6),\quad
 (\Delta_4,\Delta_6),\quad (\Delta_5,\Delta_6),\qquad
 (\Delta_1),\quad~ ...\quad~ (\Delta_7).
\eea

\section{Conversion to multiple polylogarithms}\label{sec:IntegrationAlgorithm}
\label{sect:conversion}

With the $\eps$-factorised differential equation~(\ref{Eq:DiffEquation})
at hand we may express the master integrals $\vec{g}$
in terms of iterated integrals
\bea
 I_{\omega_{i_1},\dots,\omega_{i_r}}\left[\gamma\right],
\eea
where $\omega_{i_j} \in {\mathcal A}$ and $\gamma$ is a path from a chosen boundary point to the desired point
in kinematic space.
In this section we show that all iterated integrals at any weight 
can be expressed algorithmically in terms of multiple polylogarithms.
We do this in two steps.
In the first step we perform a change of variables from our original kinematic variables $\vec{v}$
to a new set of variables obtained from a momentum twistor parameterization.
This will rationalise three of the seven square roots.
We then fix an integration path.
The remaining non-rationalised square roots are associated with triangle integrals and in any given iterated integral
there will occur at most a single non-rationalised square root.
In the second step we use different re-parameterisations of the same integration path to rationalise
the remaining square roots.

\subsection{Step 1: Momentum twistor parameterization}

In the first step we express with the help of momentum twistors (see \appref{sec:MomentumTwistorParametrization} for details)
the kinematic variables
$\vec{v}=\{s_{12},s_{23},s_{34},s_{45},s_{15},m_1^2,m_2^2,m_3^2\}$ in terms of new variables  $\pmb{x}=\{x_1,x_2,...,x_8\}$.
The transformation is given by
\bea
m_1^2&=&\frac{\mu^2}{x_4}, \nn
m_2^2&=&-\frac{x_1 (1 + x_6)}{(-1 + x_1) x_6 -x_1 (-1 + x_2) x_7} \mu^2, \nn
m_3^2&=&-\frac{N_3}{(x_5-x_7)\bigl[x_1 (-1 + x_2) x_7-(-1 + x_1) x_6\bigr]\bigl[x_4 + (-1 + x_2) x_8\bigr] \bigl[-1 +x_7-x_6 (1 + x_8)\bigr]} \mu^2, \nn
s_{12}&=&\frac{\bigl[-1 + x_1 (1 + x_4)\bigr](-1 + x_7)}{x_4 \bigl[(-1 + x_1) x_6-x_1 (-1 + x_2) x_7\bigr]} \mu^2,  \nn
s_{23}&=&-\frac{(-1+x_5-x_6) (1+x_6-x_7) (1 + x_8)^2}{(x_5-x_7)\bigl[x_4 + (-1 + x_2) x_8\bigr]\bigl[-1 + x_7 -x_6 (1 + x_8)\bigr]} \mu^2, \nn
s_{34}&=&-\frac{N_{34}}{x_4(x_5 - x_7)\bigl[x_1 (-1 + x_2) x_7-(-1 + x_1) x_6\bigr] \bigl[-1 + x_7 - x_6 (1 + x_8)\bigr]} \mu^2,  \nn        
s_{45}&=&\frac{(-1+x_7)x_8\bigl[-(-1+x_2)(x_5-x_7)x_8 - 
    x_4\bigl(1-x_7+x_8-x_5 x_8+x_6(1+x_8)\bigr)\bigr]}{x_4 (x_5 - x_7)\bigl[x_4 + (-1 + x_2) x_8\bigr]\bigl[-1+x_7-x_6 (1+x_8)\bigr]} \mu^2,  \nn
s_{15}&=&\frac{(x_3-x_5) (1+x_6-x_7)(1+x_8)}{x_4 (x_5-x_7)\bigl[-1 + x_7 - x_6 (1 + x_8)\bigr]} \mu^2,
\Label{Eq:KinVariableTransform-1}
\eea
where $N_3$ and $N_{34}$ are shorthands for
\bea
N_3&=&(-1 +x_7)\Bigl[(-1 + x_5 - x_6) (1 + x_6 - x_7) (1 + x_8)  \nn
&&~~~~~~~~~~~~~~~~+x_1(1 + x_6) (1 + x_8 + x_6 (1 + x_8) + 
              x_7 (-1 + (-2 + x_2 - x_4) x_8)) \nn 
&&~~~~~~~~~~~~~~~~-x_1 x_5 \bigl[1 + x_6 - x_7 + x_2 x_8 + x_2 x_6 x_8 - 
              x_7 x_8 + x_4 \bigl(1 + x_6 - x_7 (1 + x_8)\bigr)\bigr]\Bigr], \nn
N_{34}&=&(-1+x_7)\Bigl[-\bigl[-1 + x_1 (1 + x_4)\bigr] \bigl[(1 + x_6) x_7 x_8 + 
             x_5\bigl(1 + x_6 - x_7 (1 + x_8)\bigr)\bigr]  \nn
&&~~~~~~~~~~~~~~~~+x_3 (-1 - x_6 + x_7 - x_8 + x_5 x_8 - x_6 x_8)  \nn
&&~~~~~~~~~~~~~~~~+ x_1 x_3 \bigl[1 + x_8 - x_2 x_5 x_8 + x_6 (1 + x_8) + x_7\bigl(-1 + (-1 + x_2) x_8\bigr)\bigr]\Bigr]. 
\eea
It is easily checked that this defines for generic values of $\vec{v}$ (or $\pmb{x}$) an invertible transformation.
The above transformation has the following properties:
\begin{enumerate}
\item The square roots $\Delta_5$, $\Delta_6$ and $\Delta_7$ are rationalized.
\item The remaining square roots $\Delta_1$ - $\Delta_4$ are of degree $2$ in $x_1$.
\item Restricted to the hypersurface $x_1=0$, the remaining square roots $\Delta_1$ - $\Delta_4$ are rationalized on this hypersurface.
\end{enumerate}
We now fix the integration path to be the piecewise smooth path $\gamma=\gamma_{\mathrm{standard}}$ consisting of a straight line from 
$(0,0,\dots,0,0)$ to $(0,0,\dots,0, x_8)$, followed by a straight line from 
$(0,0,\dots,0,x_8)$ to $(0,0,\dots,x_7,x_8)$, and the pattern continues in this way.
The last segment is given by straight line from $(0,x_2,\dots,x_7,x_8)$ to 
$(x_1, x_2, \dots, x_7, x_8)$.
Due to property $3$ from the list above we will encounter the remaining square roots only on the last segment from
$(0,x_2,\dots,x_7,x_8)$ to 
$(x_1, x_2, \dots, x_7, x_8)$.
The integration over all other segments will yield multiple polylogarithms in a straightforward way.

\subsection{Step 2: Re-parameterisation of the integration path}

It remains to consider the last integration segment
from $(0,x_2,\dots,x_7,x_8)$ to $(x_1, x_2, \dots, x_7, x_8)$.
Our task is to convert the iterated integrals along this segment to multiple polylogarithms.
The remaining non-rationalized square roots are $\Delta_1, \Delta_2, \Delta_3$ and $\Delta_4$.
These square roots are associated to triangle sub-sectors and from discussion at the end of 
section~\ref{sec:alphabet} it follows that there will be never two or more of them in any given iterated
integral.
Hence, any given iterated integral will have at most one square root.
The variables $x_2, \dots, x_8$ are constant along the last integration segment.
This implies in particular, that square roots, which only depend on the variables $x_2, \dots, x_8$,
but not on $x_1$ are unproblematic.
In the variables $\pmb{x}$ the square roots $\Delta_1-\Delta_4$ factor into square roots of perfect squares
(which are un-problematic) and a square roots of a degree $2$ polynomial in $x_1$.
The last square root can always be rationalized by a re-definition of the $x_1$-variable.
As we do not change $x_2, \dots, x_8$, which are kept as constant parameters, this corresponds to a
re-parameterisation of the integration path.

Let us assume that we would like to rationalize the square root
\bea
\label{example_quadratic_square_root}
 \sqrt{ \left(x_1-a\right)\left(x_1-b\right)},
\eea
where $a$ and $b$ are functions of the variables $x_2, \dots, x_8$ and may contain square roots which only
depend on $x_2, \dots, x_8$.
The substitution
\bea
\label{substitution_x_to_t}
 x_1 & = & \frac{b t \left[ \left(a-b\right)t + 4 a b\right]}{\left(a-b\right)t^2 + 4 a b t + 4 a b^2}
\eea
rationalizes this square root with respect to the new variable $t$.
The transformation will produce $\sqrt{ab}$, but this is un-problematic, as this quantity depends only 
on $x_2, \dots, x_8$.
The transformation in eq.~(\ref{substitution_x_to_t}) has the property that $t=0$ corresponds to $x_1=0$ and
\bea
 \left. \frac{\partial x_1}{\partial t} \right|_{t=0}& = & 1.
\eea
Expressing $t$ as a function of $x_1$ we have to choose the sign of a square root.
We always choose the sign such that $x_1=0$ corresponds to $t=0$.
For the case at hand we have four remaining square roots of the type as in eq.~(\ref{example_quadratic_square_root})
and at most only one square root occurs in any given iterated integral.
We introduce four new variables $t_1$, $t_2$, $t_3$ or $t_4$ in the way discussed above 
such that $t_i$ rationalizes the square root $\Delta_i$.
Hence, we obtain five different parametrizations of the same integration path, which we will use as shown in table~\ref{table_parameterizations}.
\begin{table}
\label{table_parameterizations}
\begin{center}
\begin{tabular}{c|l}
 square roots & variables \\
 \hline
 no square root & $\{x_1,x_2,\dots,x_8\}$ \\
 $\Delta_1$ & $\{t_1,x_2,\dots,x_8\}$ \\
 $\Delta_2$ & $\{t_2,x_2,\dots,x_8\}$ \\
 $\Delta_3$ & $\{t_3,x_2,\dots,x_8\}$ \\
 $\Delta_4$ & $\{t_4,x_2,\dots,x_8\}$ \\
\end{tabular}
\end{center}
\caption{Parameterizations of the (same) integration path.
If no square roots occurs in an iterated integral, the first parameterization is used.
If the square root $\Delta_i$ occurs, the parameterization with the variable $t_i$ is used.}
\end{table}
The explicit expressions for the transformations from $x_1$ to the $t_i$'s are rather long and
may be found in the ancillary file attached to the arxiv version of this article.
By using the appropriate parameterization we ensure that we only encounter dlog-forms with rational arguments, 
hence all iterated integrals can be expressed in terms of multiple polylogarithms.

\section{Results}
\label{sect:results}

With the methods discussed above we have expressed up to weight four all master integrals for
the one-loop pentagon integral with three adjacent massive external legs
in terms of multiple polylogarithms. 
These include the lower-point integrals whose analytic $\varepsilon$-expansions have been discussed extensively in the literature.
We note that the lower-point integrals include in particular 
the 1-loop box with all external legs massive. 
For simplicity we focus on the Euclidean region.
The boundary values have been obtained with the help of the \texttt{AMFlow} package \cite{Liu:2022chg,Liu:2017jxz,Liu:2022mfb}
and the PSLQ algorithm \cite{Ferguson:1979,Ferguson:1992,Ferguson:1999,Bailey:1999nv}.
The explicit expressions for the master integrals in terms of multiple polylogarithms are rather long
and given up to weight four in the ancillary file attached to the arxiv version of this article.

As a cross-check we evaluate the master integrals at a point $\vec{v}_0$ given by
\begin{alignat}{4}
m_1^2&=-\frac{50}{51}&\quad& 
m_2^2=-\frac{175}{3798}&\quad&
m_3^2=-\frac{1102585325}{883407397698}&\quad&  
s_{12}=-\frac{67855}{129132},  \nn
s_{23}&=-\frac{55296000}{232598051}&\quad&  
s_{34}=-\frac{795086350}{2944306449}&\quad&
s_{45}=-\frac{1898617750}{11862500601}&\quad& 
s_{15}=-\frac{4800}{12019}. \Label{Eq:BoundaryPoint-1}
\end{alignat}
In terms of the momentum twistor variables defined in appendix~\ref{sec:MomentumTwistorParametrization}, this point corresponds to
\begin{alignat}{4}
x_1&=-\frac{7}{20} &\qquad& x_2=-\frac{3}{5} &\qquad& x_3=-\frac{19}{50} &\qquad& x_4=-\frac{51}{50}, \nn
x_5&=\frac{11}{5} &\qquad& x_6=-\frac{6}{5} &\qquad& x_7=\frac{9}{50} &\qquad& x_8=95, \Label{Eq:BoundaryPoint-2}
\end{alignat}
and $\mu^2=1$.
The numerical results for $g_{18}$ and $g_{19}$ up to weight four read
\bea
g_{18}&=&
7.8968007697753055...\,\epsilon^2+48.74234942348662...\,\epsilon^3+166.59728539785323...\,\epsilon^4 
\nn
g_{19}&=&
-1.2884632786357...\,\epsilon^3-6.1292127077565...\,\epsilon^4.
\eea
We verified that our result agrees with \texttt{AMFlow}.

\section{Conclusions}
\label{sec:conclusions}

In this work we considered families of Feynman integrals, which have an
$\varepsilon$-factorised differential equation which contains only dlog-forms with algebraic arguments
and where the algebraic part is given by (multiple) square roots.
These families occur frequently in precision calculations.
It is well-known that if all square roots are simultaneously rationalisable, the Feynman integrals can be expressed
in terms of multiple polylogarithms. This is a sufficient, but not a necessary criterium.
In this paper we presented weaker requirements.
We review path independence and parameterisation independence of iterated integrals.
In the context of Feynman integrals we may divide the full result into smaller path-independent subsets,
and we may use different rationalisations in different subsets.
Subsets, which are naturally path-independent, are characterised by proposition~\ref{proposition_path_independence}.
Secondly, any iterated integral is invariant under re-parameterisation of the same integration path.
Hence, we may use within the same subset different rationalisations, if they correspond to different parameterisations of the same
integration path.
We presented a non-trivial example, the one-loop pentagon function with three adjacent massive external legs involving seven square roots,
where this technique can be used to express the result in terms of multiple polylogarithms.
The technique will be useful for other Feynman integrals as well. For example, we expect that it will be useful in an attempt 
to prove that any 1-loop generic Feynman integral evaluates to multiple polylogarithms to all orders in the dimensional regulator, starting from the known canonical differential equations for these integrals~\cite{Abreu:2017mtm, Chen:2022fyw,Dlapa:2023cvx}. Furthermore, this approach may potentially be extended to higher-loop cases \cite{Abreu:2020jxa,Henn:2021cyv,Abreu:2024yit,Abreu:2024fei,Henn:2025xrc}.

\section*{Acknowledgements}
The work of GP was supported by the UK Science and Technology Facilities Council grant ST/Z001021/1. KW is very grateful for the hospitality provided by the Institut für Physik, Mainz. KW is supported by the Helmholtz-OCPC International Postdoctoral Exchange Fellowship Program. 
YZ thanks the support from the NSF of China through Grant No. 12075234, 12247103, and 12047502. This research was supported in part by the Munich Institute for Astro-, Particle and BioPhysics (MIAPbP) which is funded by the Deutsche Forschungsgemeinschaft (DFG, German Research Foundation) under Germany´s Excellence Strategy – EXC-2094 – 390783311. We  also thank the Galileo Galilei Institute for Theoretical Physics for the hospitality and the INFN for partial support during the completion of this work.

\appendix

\section{Momentum twistor parametrization}\label{sec:MomentumTwistorParametrization}

In this appendix we provide details on the momentum twistor parametrization \cite{Hodges:2005bf,Hodges:2005aj,Hodges:2009hk,Mason:2009sa,Weinzierl:2016bus,Gehrmann:2018yef}. 
The five external momenta of the pentagon integral shown in \figref{Fig:MassivePentagon} can be expressed by eight massless momenta
\bea
p_1=q_1+q_2,\quad p_2=q_3+q_4,\quad p_3=q_5+q_6,\quad p_4=q_7,\quad p_5=q_8.
\eea
where the $p_i$'s denote the original external momenta of the pentagon integral, while the $q_i$'s denote massless momenta: $q_i^2=0$ ($i=1,...,8$). 
For the momenta $q_i$ we introduce the dual coordinates as $y_{i+1}-y_i=q_i$ with $y_{i+8}=y_i$,\footnote{The standard notation of dual coordinates is $x_i$, but we use $y_i$ to avoid the confusion with the free variables in $Z$.}
The eight kinematic variables in $\vec{v}$ of the pentagon integral (\ref{Eq:PentagonIntegral}) can then be expressed as follows: 
\bea
s_{12}&=&y^2_{1,5},\quad s_{23}=y_{3,7}^2,\quad s_{34}=y_{5,8}^2, \quad s_{45}=y_{1,7}^2, \nn
s_{15}&=&y^2_{3,8},\quad m_1^2=y_{1,3}^2,\quad m_2^2=y_{3,5}^2,\quad m_3^2=y_{5,7}^2,
\eea
where $y_{i,j}=y_j-y_i$. 
Although we work with dimensional regularisation,
the external momenta lie always in a four-dimensional sub-space.
We may therefore use eight momentum twistor variables $Z_i=(\lambda_i,\mu_i)$ to encode the corresponding $q_i$ ($i=1,...,8$).
The $y_{i,j}^2$ have the following expressions in terms of the momentum twistor variables:
\bea
\label{y2_in_twistors}
y_{i,j}^2=\frac{\langle i-1 i j-1 j\rangle}{\langle i-1\,i\, I_{\infty}\rangle\,\langle j-1\,j\,I_{\infty}\rangle} \mu^2.
\eea
Here, the four bracket $\langle ijkl \rangle$ is defined as the determinant of four momentum twistors. Explicitly it is given by
\bea
\langle ijkl\rangle=\epsilon^{ijkl}Z_iZ_jZ_kZ_l.
\eea
The $I_{\infty}$ refers to two auxiliary momentum twistors $Z_9$ and $Z_{10}$: $I_{\infty}=(Z_9,Z_{10})$. 
It is convenient to introduce the scale $\mu$ in eq.~(\ref{y2_in_twistors}), this ensures that the twistors are dimensionless. 
The momentum twistor parameterization is highly redundant and we are allowed to make specific choices to reduce 
this redundancy.
Specifically, we choose the following 
momentum twistor parametrization in our calculation
\bea
\left(Z_1,Z_2,Z_3,Z_4,Z_5,Z_6,Z_7,Z_8,Z_9,Z_{10}\right)=\left(
\begin{array}{cccccccccc}
 1 & 0 & 0 & 0 & x_6+1 & 1 & x_5 & x_3 & 1 & 0 \\
 0 & 1 & 0 & \frac{1}{x_1} & x_2 & z_{2,6} & 1 & x_4+1 & 0 & 1 \\
 0 & 0 & 1 & x_7 & 1 & z_{3,6} & z_{3,7} & x_7 & 1 & 0 \\
 0 & 0 & 0 & 1 & 1 & 1 & 1 & 1 & 0 & 1 \\
\end{array}
\right), \Label{Eq:MomentumTwistorParametrization}
\eea
where
\bea
z_{2,6}&=&\frac{x_2 x_5 x_8-x_2 x_8+x_4 x_5-x_4+x_5-x_6 x_8-x_6-1}{(x_8+1) (x_5-x_6-1)}, \nn
z_{3,6}&=&-\frac{x_6 x_7 x_8+x_6 x_7-x_7^2-x_7 x_8+x_7+x_8}{(x_8+1) (-x_6+x_7-1)}, \nn
z_{3,7}&=&-\frac{-x_5 x_7 x_8+x_5 x_8+x_6 x_7 x_8+x_6 x_7-x_7^2+x_7}{(x_8+1) (-x_6+x_7-1)},
\eea
in which $x_1,...,x_8$ are free variables. 
Note that \eqref{Eq:KinVariableTransform-1} is the explicit expression for the parametrization in \eqref{Eq:MomentumTwistorParametrization}. 
It is easily checked that this defines for generic values of $\vec{v}$ an invertible transformation between 
the eight kinematic variables of $\vec{v}$ and $x_1,\dots,x_8$.
While some matrix entries in $Z$ are not so simple rational expressions of these variables, this particular choice has the benefit that it rationalizes three square roots, $\Delta_5,\Delta_6$ and $\Delta_7$, whereas a generic momentum twistor parametrization only rationalizes $\Delta_7$.

\section{The $R_i$ terms of the letters}\label{sec:R-termsInLetter}

All the $R_i$ that are shown in the letters in \eqref{Eq:PentagonLetter} have the explicit expressions as follows:
\bea
R_{28}&=&-s_{34} m_1^4-2 m_2^2 s_{12} m_1^2+s_{12} s_{15} m_1^2+m_2^2 s_{34} m_1^2+s_{12} s_{34} m_1^2-s_{12}^2 s_{15}+m_2^2 s_{12} s_{15},  \nn
R_{29}&=&-s_{15} m_3^4+m_2^2 s_{15} m_3^2-2 m_2^2 s_{23} m_3^2+s_{15} s_{23} m_3^2+s_{23} s_{34} m_3^2-s_{23}^2 s_{34}+m_2^2 s_{23} s_{34}, \nn
R_{32}&=&-s_{12} s_{15}^2+m_2^2 s_{12} s_{15}+m_1^2 s_{34} s_{15}-2 m_2^2 s_{34} s_{15}+s_{12} s_{34} s_{15}-m_1^2 s_{34}^2+m_1^2 m_2^2 s_{34}, \nn
R_{33}&=&-m_3^2 s_{15}^2+m_2^2 m_3^2 s_{15}-2 m_2^2 s_{34} s_{15}+m_3^2 s_{34} s_{15}+s_{23} s_{34} s_{15}-s_{23} s_{34}^2+m_2^2 s_{23} s_{34}, \nn
R_{34}&=&m_3^4 m_1^4+s_{34}^2 m_1^4-2 m_3^2 s_{34} m_1^4-2 s_{23} s_{34}^2 m_1^2-2 m_3^4 s_{15} m_1^2+2 m_3^2 s_{12} s_{15} m_1^2-2 m_3^2 s_{12} s_{23} m_1^2 \nn
&&~~~+2 m_3^2 s_{15} s_{34} m_1^2-2 s_{12} s_{15} s_{34} m_1^2+2 m_3^2 s_{23} s_{34} m_1^2+2 s_{12} s_{23} s_{34} m_1^2-2 s_{34}^2 s_{45} m_1^2 \nn
&&~~~-2 m_2^2 m_3^2 s_{45} m_1^2+2 m_3^2 s_{15} s_{45} m_1^2+2 m_2^2 s_{34} s_{45} m_1^2+2 m_3^2 s_{34} s_{45} m_1^2+2 s_{15} s_{34} s_{45} m_1^2 \nn
&&~~~-4 s_{23} s_{34} s_{45} m_1^2+m_3^4 s_{15}^2+s_{12}^2 s_{15}^2-2 m_3^2 s_{12} s_{15}^2+s_{12}^2 s_{23}^2+s_{23}^2 s_{34}^2+m_2^4 s_{45}^2+s_{15}^2 s_{45}^2 \nn
&&~~~+s_{34}^2 s_{45}^2-2 m_2^2 s_{15} s_{45}^2-2 m_2^2 s_{34} s_{45}^2-2 s_{15} s_{34} s_{45}^2-2 s_{12}^2 s_{15} s_{23}+2 m_3^2 s_{12} s_{15} s_{23}-2 s_{12} s_{23}^2 s_{34} \nn
&&~~~-2 m_3^2 s_{15} s_{23} s_{34}+2 s_{12} s_{15} s_{23} s_{34}-2 m_3^2 s_{15}^2 s_{45}-2 s_{12} s_{15}^2 s_{45}-2 s_{23} s_{34}^2 s_{45}+2 m_2^2 m_3^2 s_{15} s_{45} \nn
&&~~~+2 m_2^2 s_{12} s_{15} s_{45}-4 m_3^2 s_{12} s_{15} s_{45}-2 m_2^2 s_{12} s_{23} s_{45}+2 s_{12} s_{15} s_{23} s_{45}-4 m_2^2 s_{15} s_{34} s_{45} \nn
&&~~~+2 m_3^2 s_{15} s_{34} s_{45}+2 s_{12} s_{15} s_{34} s_{45}+2 m_2^2 s_{23} s_{34} s_{45}+2 s_{12} s_{23} s_{34} s_{45}+2 s_{15} s_{23} s_{34} s_{45}, \nn
R_{35}&=&-m_3^4 m_1^2+m_2^2 m_3^2 m_1^2-m_2^2 s_{12} m_1^2+m_3^2 s_{12} m_1^2+m_3^2 s_{23} m_1^2+s_{12} s_{23} m_1^2+m_2^2 s_{45} m_1^2+m_3^2 s_{45} m_1^2 \nn
&&~~~-s_{23} s_{45} m_1^2-s_{12} s_{23}^2-m_2^2 s_{45}^2-m_1^4 m_3^2-s_{12}^2 s_{23}-m_2^2 m_3^2 s_{23}+m_2^2 s_{12} s_{23}+m_3^2 s_{12} s_{23} \nn
&&~~~-m_2^4 s_{45}+m_2^2 m_3^2 s_{45}+m_2^2 s_{12} s_{45}-m_3^2 s_{12} s_{45}+m_2^2 s_{23} s_{45}+s_{12} s_{23} s_{45}, \nn
R_{44}&=&m_2^2 m_3^2 m_1^2-2 m_2^2 s_{12} m_1^2+m_3^2 s_{12} m_1^2+s_{12} s_{23} m_1^2+m_2^2 s_{45} m_1^2-m_1^4 m_3^2-s_{12}^2 s_{23}+m_2^2 s_{12} s_{23} \nn
&&~~~-m_2^4 s_{45}+m_2^2 s_{12} s_{45}, \nn
R_{45}&=&-s_{45} m_2^4+m_1^2 m_3^2 m_2^2-2 m_3^2 s_{23} m_2^2+s_{12} s_{23} m_2^2+m_3^2 s_{45} m_2^2+s_{23} s_{45} m_2^2-m_1^2 m_3^4-s_{12} s_{23}^2 \nn
&&~~~+m_1^2 m_3^2 s_{23}+m_3^2 s_{12} s_{23}, \nn
R_{46}&=&-m_1^2 m_3^4+m_1^2 s_{12} m_3^2+s_{12} s_{23} m_3^2+m_1^2 s_{45} m_3^2+m_2^2 s_{45} m_3^2-2 s_{12} s_{45} m_3^2-m_2^2 s_{45}^2-s_{12}^2 s_{23} \nn
&&~~~+m_2^2 s_{12} s_{45}+s_{12} s_{23} s_{45}, \nn
R_{47}&=&m_3^2 s_{23} m_1^2+s_{12} s_{23} m_1^2+m_2^2 s_{45} m_1^2+m_3^2 s_{45} m_1^2-2 s_{23} s_{45} m_1^2-s_{12} s_{23}^2-m_2^2 s_{45}^2-m_1^4 m_3^2 \nn
&&~~~+m_2^2 s_{23} s_{45}+s_{12} s_{23} s_{45}, \nn
R_{48}&=&-m_1^2 m_3^2+s_{15} m_3^2+s_{12} s_{15}-s_{12} s_{23}+m_1^2 s_{34}-2 s_{15} s_{34}+s_{23} s_{34}+m_2^2 s_{45}-s_{15} s_{45}-s_{34} s_{45}, \nn
R_{49}&=&-s_{34} m_1^4+m_2^2 m_3^2 m_1^2-2 m_2^2 s_{12} m_1^2+m_3^2 s_{12} m_1^2+m_3^2 s_{15} m_1^2+s_{12} s_{15} m_1^2+s_{12} s_{23} m_1^2+m_2^2 s_{34} m_1^2 \nn
&&~~~-2 m_3^2 s_{34} m_1^2+s_{12} s_{34} m_1^2+s_{23} s_{34} m_1^2+m_2^2 s_{45} m_1^2-s_{15} s_{45} m_1^2+s_{34} s_{45} m_1^2-m_1^4 m_3^2-s_{12}^2 s_{15} \nn
&&~~~-m_2^2 m_3^2 s_{15}+m_2^2 s_{12} s_{15}+m_3^2 s_{12} s_{15}-s_{12}^2 s_{23}+m_2^2 s_{12} s_{23}-2 s_{12} s_{15} s_{23}-m_2^2 s_{23} s_{34}+s_{12} s_{23} s_{34} \nn
&&~~~-m_2^4 s_{45}+m_2^2 s_{12} s_{45}+m_2^2 s_{15} s_{45}+s_{12} s_{15} s_{45}+m_2^2 s_{34} s_{45}-s_{12} s_{34} s_{45}, \nn
R_{50}&=&-m_3^2 s_{15}^2-s_{12} s_{15}^2+s_{45} s_{15}^2+m_1^2 m_3^2 s_{15}+s_{12} s_{23} s_{15}+m_1^2 s_{34} s_{15}+s_{23} s_{34} s_{15}-m_2^2 s_{45} s_{15} \nn
&&~~~-s_{34} s_{45} s_{15}-2 m_1^2 s_{23} s_{34}, \nn
R_{51}&=&m_3^2 m_1^4-s_{34} m_1^4-2 m_3^2 s_{15} m_1^2+s_{12} s_{15} m_1^2-s_{12} s_{23} m_1^2+s_{15} s_{34} m_1^2+s_{23} s_{34} m_1^2-m_2^2 s_{45} m_1^2+s_{15} s_{45} m_1^2 \nn
&&~~~+s_{34} s_{45} m_1^2+m_3^2 s_{15}^2-s_{12} s_{15}^2+s_{12} s_{15} s_{23}-s_{15} s_{23} s_{34}-s_{15}^2 s_{45}+m_2^2 s_{15} s_{45}-2 s_{12} s_{15} s_{45}+s_{15} s_{34} s_{45}, \nn
R_{52}&=&-m_3^2 s_{15}^2+s_{12} s_{15}^2-s_{45} s_{15}^2+m_1^2 m_3^2 s_{15}+m_3^2 s_{23} s_{15}-2 s_{12} s_{23} s_{15}-m_1^2 s_{34} s_{15}+s_{23} s_{34} s_{15}+m_2^2 s_{45} s_{15} \nn
&&~~~-2 m_3^2 s_{45} s_{15}+s_{23} s_{45} s_{15}+s_{34} s_{45} s_{15}+s_{12} s_{23}^2-m_1^2 m_3^2 s_{23}-s_{23}^2 s_{34}+m_1^2 s_{23} s_{34}-m_2^2 s_{23} s_{45}+s_{23} s_{34} s_{45}, \nn
R_{53}&=&-s_{45} m_2^4+m_1^2 m_3^2 m_2^2+m_3^2 s_{15} m_2^2-s_{12} s_{15} m_2^2-2 m_3^2 s_{23} m_2^2+s_{12} s_{23} m_2^2-m_1^2 s_{34} m_2^2+s_{23} s_{34} m_2^2 \nn
&&~~~+m_3^2 s_{45} m_2^2+s_{15} s_{45} m_2^2+s_{23} s_{45} m_2^2+s_{34} s_{45} m_2^2-m_1^2 m_3^4-s_{12} s_{23}^2-m_3^4 s_{15}-2 m_1^2 m_3^2 s_{15} \nn
&&~~~+m_3^2 s_{12} s_{15}+m_1^2 m_3^2 s_{23}+m_3^2 s_{12} s_{23}+m_3^2 s_{15} s_{23}+s_{12} s_{15} s_{23}-s_{23}^2 s_{34}+m_1^2 m_3^2 s_{34}+m_1^2 s_{23} s_{34} \nn
&&~~~+m_3^2 s_{23} s_{34}-2 s_{12} s_{23} s_{34}+m_3^2 s_{15} s_{45}-s_{15} s_{23} s_{45}-m_3^2 s_{34} s_{45}+s_{23} s_{34} s_{45}, \nn
R_{54}&=&s_{45} m_2^4-m_1^2 m_3^2 m_2^2+m_3^2 s_{15} m_2^2+s_{12} s_{15} m_2^2-s_{12} s_{23} m_2^2+m_1^2 s_{34} m_2^2-2 s_{15} s_{34} m_2^2+s_{23} s_{34} m_2^2 \nn
&&~~~-2 s_{15} s_{45} m_2^2-2 s_{34} s_{45} m_2^2-m_3^2 s_{15}^2-s_{12} s_{15}^2-m_1^2 s_{34}^2-s_{23} s_{34}^2+m_1^2 m_3^2 s_{15}-2 m_3^2 s_{12} s_{15}+s_{12} s_{15} s_{23} \nn
&&~~~+m_1^2 m_3^2 s_{34}+m_1^2 s_{15} s_{34}+m_3^2 s_{15} s_{34}+s_{12} s_{15} s_{34}-2 m_1^2 s_{23} s_{34}+s_{12} s_{23} s_{34}+s_{15} s_{23} s_{34}+s_{15}^2 s_{45} \nn
&&~~~+s_{34}^2 s_{45}-2 s_{15} s_{34} s_{45}, \nn
R_{55}&=&-m_1^2 m_3^4+s_{15} m_3^4+m_1^2 s_{12} m_3^2-2 s_{12} s_{15} m_3^2+s_{12} s_{23} m_3^2+m_1^2 s_{34} m_3^2-s_{23} s_{34} m_3^2+m_1^2 s_{45} m_3^2+m_2^2 s_{45} m_3^2 \nn
&&~~~-2 s_{12} s_{45} m_3^2-2 s_{15} s_{45} m_3^2+s_{34} s_{45} m_3^2-m_2^2 s_{45}^2+s_{15} s_{45}^2-s_{34} s_{45}^2+s_{12}^2 s_{15}-s_{12}^2 s_{23}-m_1^2 s_{12} s_{34} \nn
&&~~~+s_{12} s_{23} s_{34}+m_2^2 s_{12} s_{45}-2 s_{12} s_{15} s_{45}+s_{12} s_{23} s_{45}+m_1^2 s_{34} s_{45}-2 m_2^2 s_{34} s_{45}+s_{12} s_{34} s_{45}+s_{23} s_{34} s_{45}, \nn
R_{56}&=&s_{34} m_1^4+m_3^2 s_{15} m_1^2-s_{12} s_{15} m_1^2+m_3^2 s_{23} m_1^2+s_{12} s_{23} m_1^2-2 s_{23} s_{34} m_1^2+m_2^2 s_{45} m_1^2+m_3^2 s_{45} m_1^2+s_{15} s_{45} m_1^2 \nn
&&~~~-2 s_{23} s_{45} m_1^2-2 s_{34} s_{45} m_1^2-s_{12} s_{23}^2-m_2^2 s_{45}^2-s_{15} s_{45}^2+s_{34} s_{45}^2-m_1^4 m_3^2-m_3^2 s_{15} s_{23}+s_{12} s_{15} s_{23} \nn
&&~~~+s_{23}^2 s_{34}-2 m_2^2 s_{15} s_{45}+m_3^2 s_{15} s_{45}+s_{12} s_{15} s_{45}+m_2^2 s_{23} s_{45}+s_{12} s_{23} s_{45}+s_{15} s_{23} s_{45}-2 s_{23} s_{34} s_{45}, \nn
R_{57}&=&m_3^4 m_1^4-m_3^2 s_{34} m_1^4-m_3^4 s_{15} m_1^2+m_3^2 s_{12} s_{15} m_1^2-2 m_3^2 s_{12} s_{23} m_1^2+m_3^2 s_{23} s_{34} m_1^2+s_{12} s_{23} s_{34} m_1^2 \nn
&&~~~-2 m_2^2 m_3^2 s_{45} m_1^2+m_3^2 s_{15} s_{45} m_1^2+m_2^2 s_{34} s_{45} m_1^2+m_3^2 s_{34} s_{45} m_1^2-2 s_{23} s_{34} s_{45} m_1^2+s_{12}^2 s_{23}^2+m_2^4 s_{45}^2 \nn
&&~~~-m_2^2 s_{15} s_{45}^2-m_2^2 s_{34} s_{45}^2-s_{12}^2 s_{15} s_{23}+m_3^2 s_{12} s_{15} s_{23}-s_{12} s_{23}^2 s_{34}+m_2^2 m_3^2 s_{15} s_{45}+m_2^2 s_{12} s_{15} s_{45} \nn
&&-2 m_3^2 s_{12} s_{15} s_{45}-2 m_2^2 s_{12} s_{23} s_{45}+s_{12} s_{15} s_{23} s_{45}+m_2^2 s_{23} s_{34} s_{45}+s_{12} s_{23} s_{34} s_{45}.
\eea


\section{Supplementary material}
\label{sect:supplement}

Attached to the arxiv version of this article are the following auxiliary files in {\tt Mathematica} syntax.
\begin{itemize}
\item \texttt{pent3mAB.m:} Alphabet and square roots.

\item \texttt{UTMatrix.m:} Matrix of canonical differential equations.

\item \texttt{kinvecReplace2.m:} Transformation between kinematic variables and $\{x_1,\dots,x_8\}$.

\item \texttt{sDpentReplace2.m:} Rationalized expressions of square roots: $\Delta_1,\dots,\Delta_7$.

\item \texttt{x1Replace.m:} Mapping between $x_1$ and $t_1,....,t_4$.

\item \texttt{FeynIntegralwithBoundaryweight1-3.m:} The results of 19 canonical integrals up to weight-3.

\item \texttt{FeynIntegralwithBoundaryweight4.m:} Pentagon integral of weight-4.

\item \texttt{lettertPentvector2.m:} Arguments corresponding to variables $\{x_1,t_1,\dots,t_4\}$ in the MPLs.

\item \texttt{letterkinvecPent2.m:} Arguments corresponding to variables $x_2,\dots,x_8$ in the MPLs.
\end{itemize}

\bibliographystyle{JHEP}
\bibliography{reference}

\end{document}